\documentclass[conference]{IEEEtran}
\usepackage{times}

\usepackage{caption}
\ifCLASSINFOpdf
\else
\fi
\usepackage[cmex10]{amsmath}
\usepackage{amssymb}
\usepackage{amsfonts}
\usepackage{amscd}
\usepackage{lmodern}

\usepackage{array}
\usepackage{url}
\usepackage{graphicx}
\usepackage{mathtools}
\usepackage{etex}

\usepackage[subpreambles]{standalone}
\usepackage{amssymb}
\usepackage{mathtools}
\usepackage{algorithm}
\usepackage{bm}
\usepackage{graphicx}
\usepackage{url}
\usepackage{tikz}
\usetikzlibrary{arrows,matrix,positioning}
\usepackage{algorithm}
\usepackage{rotating}
\usepackage[english]{babel} 
\usepackage{color,array}
\usepackage{enumerate}
\usepackage{bm}
\usepackage{graphicx}
\usepackage[caption=false]{subfig}
\usepackage{url}
\usepackage{algpseudocode}
\usepackage{standalone}
\usepackage{tikz}
\usetikzlibrary{arrows,matrix,positioning,patterns}
\usepackage{pgfplots}
\usetikzlibrary{matrix,positioning,decorations.pathreplacing,decorations.pathmorphing,backgrounds,chains}
\usepackage{relsize}
\colorlet{boxgray}{lightgray!20}

\newtheorem{definition}{Definition}
\newtheorem{theorem}{Theorem}

\newtheorem{proposition}{Proposition}

\newtheorem{example}{Example}

\newtheorem{proof}{Proof}

\newsavebox{\ieeealgbox}

\newcommand{\shorten}[1]{}

\begin{document}
\pagenumbering{gobble}
\bstctlcite{IEEEexample:BSTcontrol}
%
%3804
%7D6F8251
\title{On the Construction of Quasi-Binary and Quasi-Orthogonal Matrices over Finite Fields}
%%
%APPROACHING MAXIMUM EMBEDDING EFFICIENCY ON SMALL COVERS USING STAIRCASE-GENERATOR CODES
%\vspace{-.5cm}}
%% conference papers do not typically use \thanks and this command
%% is locked out in conference mode. If really needed, such as for
%% the acknowledgment of grants, issue a \IEEEoverridecommandlockouts

\author{\IEEEauthorblockN{Danilo Gligoroski\IEEEauthorrefmark{1}
and
Kristian Gj\o{}steen\IEEEauthorrefmark{2}
and 
Katina Kralevska\IEEEauthorrefmark{1} 
}
\IEEEauthorblockA{
\IEEEauthorrefmark{1}Department of Information Security and Communication Technologies,
\IEEEauthorrefmark{2}Department of Mathematics, \\
Norwegian University of Science and Technology, 
Trondheim, Norway,\\
Email: \{danilog, kristianj, katinak\}@ntnu.no\\
}
}
\maketitle

\begin{abstract} Orthogonal and quasi-orthogonal matrices have a long history of use in digital image processing, digital and wireless communications, cryptography and many other areas of computer science and coding theory. The practical benefits of using orthogonal matrices come from the fact that the computation of inverse matrices is avoided, by simply using the transpose of the orthogonal matrix. In this paper, we introduce a new family of matrices over finite fields that we call \emph{Quasi-Binary and Quasi-Orthogonal Matrices}. We call the matrices quasi-binary due to the fact that matrices have only two elements $a, b \in \mathbb{F}_q$, but those elements are not $0$ and $1$. In addition, the reason why we call them quasi-orthogonal is due to the fact that their inverses are obtained not just by a simple transposition, but there is a need for an additional operation: a replacement of $a$ and $b$ by two other values $c$ and $d$. We give a simple relation between the values $a, b, c$ and $d$ for any finite field and especially for finite fields with characteristic 2. Our construction is based on incident matrices from cyclic Latin Rectangles and the efficiency of the proposed algorithm comes from the avoidance of matrix-matrix or matrix-vector multiplications. %We have implemented it in several computer languages such as C, SageMath (Python) and Mathematica. 
% We set up several concrete instances of our scheme with concrete parameters and we present an initial security analysis of the scheme.
\end{abstract}
\vspace{.2cm}
\begin{IEEEkeywords}
Orthogonal matrices, Quasi-orthogonal matrices, Quasi-binary matrices, Latin Rectangles \end{IEEEkeywords}
%\tableofcontents
%\newpage

\section{Introduction}
%\vspace{0.1cm}
Orthogonal matrices have been used in numerous applied fields such as signal processing \cite{horadam2007hadamard}, image processing \cite{Balonin2015}, coding theory \cite{sharma2014construction}, cryptology \cite{gligoroski2008edon}, and wireless communication \cite{adams2009journey,898239}. The practical benefits of using orthogonal matrices come from the fact that the computation of inverse matrices is avoided, simply by using the transpose of the orthogonal matrix. 

Probably the most well known orthogonal matrices are Hadamard matrices that were first defined by Sylvester in \cite{sylvester1867lx} and extensively studied by Hadamard in \cite{hadamard1893resolution}. The concept of orthogonality has been extended into various concepts of quasi-orthogonality when some of the conditions or properties of the matrices are relaxed or generalized (see for example \cite{Jafarkhani2001,sharma2003improved}).

A pioneering work in the orthogonal matrices over finite fields was done by MacWilliams \cite{macwilliams1971orthogonal}. That work was extended by Byrd and Vaughan in \cite{byrd1978counting} and Zhang in \cite{zhang1997construction} who proposed algorithms for constructing orthogonal circulant matrices for arbitrary dimensions $n$ in any finite field. 

\subsection{Our contribution}
In this paper, we introduce a new family of matrices over finite fields that we call \emph{Quasi-Binary and Quasi-Orthogonal Matrices} and we give an algorithm for their construction. We call the matrices \emph{quasi-binary} due to the fact that matrices have only two elements $a, b \in \mathbb{F}_q$ where $\mathbb{F}_q$ is a finite field of size $q$ and those elements are not $0$ and $1$. Additionally, we call these matrices \emph{quasi-orthogonal} because their inverses are not obtained just by a simple transposition, but there is a need for an additional operation: a replacement of $a$ and $b$ by two other values $c$ and $d$ where $c, d \in \mathbb{F}_q$.

We define an efficient algorithm for constructing quasi-binary and quasi-orthogonal matrices. The algorithm is based on incident matrices from cyclic Latin Rectangles. The algorithm is very efficient since it completely avoids operations of matrix-matrix or matrix-vector multiplications. We have implemented it in several computer languages such as C, SageMath (Python) \cite{sage} and Mathematica\footnote{As an online addition to this article, a free C source code is available and can be found as part of Edon-K and Edon-S packets in SUPERCOP cryptographic benchmarking tool \cite{SUPERCOP}}. We show several examples of these matrices.

The rest of the paper is organized as follows. Section \ref{sec:Prelim} gives some mathematical definitions and preliminaries. The main algorithm and some examples of its output are presented in Section \ref{sec:algorithm}. Section \ref{sec:conclusions} concludes the paper.

\section{Mathematical Preliminaries }\label{sec:Prelim}
%\vspace{0.3cm}
We give here some basic definitions and mathematical preliminaries about orthogonal matrices over finite fields. Readers interested in more information about this area are referred to \cite{zhang1997construction}.

\begin{definition}
	A matrix $\mathbf{P}$ is binary if its elements are either $0$ or $1$.
\end{definition}

\begin{definition}
	A square nonsingular matrix $\mathbf{P}$ is orthogonal if $\mathbf{P}^{-1} = \mathbf{P}^{T}$.
\end{definition}

\begin{definition}
	Let $\mathbf{P}$ be an $n\times n$ orthogonal binary matrix, and let $a \neq b$ are two nonzero elements from $\mathbb{F}_q$. We call the matrix $\mathbf{P}_{a, b}$ obtained from $ \mathbf{P}$ a quasi-binary quasi-orthogonal matrix if every value of $0$ is replaced by $a$ and every value of $1$ is replaced by $b$. For the inverse matrix $\mathbf{P}^{-1}_{a, b}$ it holds that $\mathbf{P}^{-1}_{a, b} = (\mathbf{P}_{a, b})^{T}_{c, d}$, i.e., every value of $a$ in the transpose matrix $(\mathbf{P}_{a, b})^{T}$ is replaced by $c$ and every value of $b$ is replaced by $d$.
	\label{Definition:QuasiBinaryQuasiOrtho}
\end{definition}

\begin{example}
	Let us work in this example in the finite field $\mathbb{F}_{16}$ with an irreducible polynomial: $x^4+x^3+1$. Next, instead of representing the elements $a \in \mathbb{F}_{16}$ as polynomials $a = a_3 x^3 + a_2 x^2 + a_1 x + a_0 $, where $a_i \in \{0, 1\}$, let us represent the elements with the corresponding integer values from their binary representation.  For example, if $a = 0\cdot x^3 + 1\cdot  x^2 + 1 \cdot x + 1 $, then we write $a = (0, 1, 1, 1)_2 = 7$. 
	
	Let us have the following binary matrix: $$\mathbf{P} = \left(
	\begin{array}{cccccccc}
	0 & 1 & 0 & 1 & 1 & 0 & 0 & 0 \\
	0 & 0 & 1 & 0 & 0 & 1 & 0 & 1 \\
	0 & 0 & 0 & 0 & 0 & 1 & 1 & 1 \\
	0 & 0 & 1 & 0 & 0 & 1 & 1 & 0 \\
	0 & 0 & 1 & 0 & 0 & 0 & 1 & 1 \\
	1 & 1 & 0 & 0 & 1 & 0 & 0 & 0 \\
	1 & 0 & 0 & 1 & 1 & 0 & 0 & 0 \\
	1 & 1 & 0 & 1 & 0 & 0 & 0 & 0 \\
	\end{array}
	\right) .$$ 
	
	It is easy to check that $\mathbf{P} \cdot \mathbf{P}^T = I$, i.e. that $\mathbf{P}$ is an orthogonal binary matrix. Let us choose the following values: $a = 7$, $b = 13$ and their corresponding values $c = 4$ and $d = 15$. Then the two quasi-binary and quasi-orthogonal matrices are:
	$$\mathbf{P}_{7, 13} = \small
	\left(
	\begin{array}{cccccccc}
	7 & 13 & 7 & 13 & 13 & 7 & 7 & 7 \\
	7 & 7 & 13 & 7 & 7 & 13 & 7 & 13 \\
	7 & 7 & 7 & 7 & 7 & 13 & 13 & 13 \\
	7 & 7 & 13 & 7 & 7 & 13 & 13 & 7 \\
	7 & 7 & 13 & 7 & 7 & 7 & 13 & 13 \\
	13 & 13 & 7 & 7 & 13 & 7 & 7 & 7 \\
	13 & 7 & 7 & 13 & 13 & 7 & 7 & 7 \\
	13 & 13 & 7 & 13 & 7 & 7 & 7 & 7 \\
	\end{array}
	\right)$$ and 
	$$\mathbf{P}_{4, 15} = \small
	\left(
	\begin{array}{cccccccc}
	4 & 4 & 4 & 4 & 4 & 15 & 15 & 15 \\
	15 & 4 & 4 & 4 & 4 & 15 & 4 & 15 \\
	4 & 15 & 4 & 15 & 15 & 4 & 4 & 4 \\
	15 & 4 & 4 & 4 & 4 & 4 & 15 & 15 \\
	15 & 4 & 4 & 4 & 4 & 15 & 15 & 4 \\
	4 & 15 & 15 & 15 & 4 & 4 & 4 & 4 \\
	4 & 4 & 15 & 15 & 15 & 4 & 4 & 4 \\
	4 & 15 & 15 & 4 & 15 & 4 & 4 & 4 \\
	\end{array}
	\right). 
	$$
	\label{Example:QuasiBinQuasiOrth}
\end{example}

Without going deeper in the Combinatorics, we give here some basic definitions and basic propositions for Latin Rectangles and incidence matrices. Proofs of these basic properties can be found for example in \cite{Stinsonbook}.

\begin{definition}
	A $k\times n$ Latin Rectangle $L=[l_{i,j}]_{k \times n}$ is a $k\times n$ array (where $k \leq n$) in which every row $R_0, \ldots, R_{k-1}$ is a permutation of an $n$-element set $X=\{0, 1, \ldots, n-1\}$, and the elements in each column $C_0, C_1, \ldots, C_{n-1}$ appear at most once. Note that zero-indexing style is used, i.e. $i, j \in \{0, 1, \ldots, n-1\}$.
	\label{Latinrectangle}
\end{definition}

\begin{definition}
	Let $L=[l_{i,j}]_{k \times n}$ be a Latin Rectangle, with columns $C_0, C_1, \ldots, C_{n-1} $. The incidence matrix of $L$ is the $n\times n$ binary matrix $\mathbf{M} = (m_{i,j})$  defined by the rule
	$$
	m_{i,j} =
	\begin{cases}
	1, & \text{if}\  i\in C_j, \nonumber\\
	0, & \text{if}\  i\notin C_j. \nonumber\\
	\end{cases}
	$$
	\label{Definition:incidencematrix}
\end{definition}

\begin{example}
	Let $k=3$ and $n=8$, and let\\ $L=[l_{i,j}]_{3 \times 8} = \left[
	\begin{array}{cccccccc}
	6 & 5 & 4 & 3 & 1 & 7 & 0 & 2 \\
	4 & 3 & 1 & 7 & 0 & 2 & 6 & 5 \\
	1 & 7 & 0 & 2 & 6 & 5 & 4 & 3 \\
	\end{array}
	\right].$ Then the corresponding incidence matrix is: $$\mathbf{M} = \left(
	\begin{array}{cccccccc}
	0 & 0 & 1 & 0 & 1 & 0 & 1 & 0 \\
	1 & 0 & 1 & 0 & 1 & 0 & 0 & 0 \\
	0 & 0 & 0 & 1 & 0 & 1 & 0 & 1 \\
	0 & 1 & 0 & 1 & 0 & 0 & 0 & 1 \\
	1 & 0 & 1 & 0 & 0 & 0 & 1 & 0 \\
	0 & 1 & 0 & 0 & 0 & 1 & 0 & 1 \\
	1 & 0 & 0 & 0 & 1 & 0 & 1 & 0 \\
	0 & 1 & 0 & 1 & 0 & 1 & 0 & 0 \\
	\end{array}
	\right).$$ 
	\label{Ex:LatSq}
\end{example}

\begin{proposition}
	The incidence matrix $M = (m_{i,j})$ of any Latin Rectangle with dimensions $k\times n$ is balanced matrix with $k$ ones in each row and each column.
	\label{Proposition:balanced}
\end{proposition}

\begin{proposition}
	Let $M = (m_{i,j})$ be the incidence matrix of a Latin Rectangle with dimensions $k\times n$ and even $n$. If $M$ is nonsingular, then $k$ is odd.
	\label{Proposition:odd}
\end{proposition}

\begin{definition}
	Let $R_0$ be a permutation of the $n$-element set $X=\{0, 1, \ldots, n-1\}$. Let $gcd(n, rot)=1$ and let $R_i = RotateLeft(R_0, i\times rot)$ are obtained with left rotation of the initial permutation $R_0$ by $i \times rot$ positions. Then the Latin Rectangle $L=[l_{i,j}]_{k \times n}$ with rows $R_0, \ldots, R_{k-1}$ is called a Cyclic Latin Rectangle. \label{Def:CyclicLatinRect}
\end{definition}

In Example \ref{Ex:LatSq}, the initial permutation is $R_0 = \{6, 5, 4, 3, 1, 7, 0, 2\}$, and $R_1$ and $R_2$ are obtained with its rotation by 2 elements to the left.

\section{Algorithm for Construction of Quasi-Binary and Quasi-Orthogonal Matrices}\label{sec:algorithm}
%\vspace{0.3cm}
Instead of using Zhang's algorithms \cite{zhang1997construction} for construction of orthogonal matrices with arbitrary dimensions $n$, we use the definitions and the properties of Latin Rectangles defined in the previous Section. Note that similar types of matrices defined by Latin Rectangles have been used in various cryptographic algorithms such as hash functions \cite{gligoroski2008edon}, multivariate signature schemes \cite{gligoroski2011mqq}, and in various codes \cite{gligoroski2014families,kralevska2014joint,kralevska2013balanced}.

\begin{table}
	\normalsize
	\begin{tabular}{cccc}
		(8,3,2)     & (12,3,3)    & (16,3,4)    & (18,5,3)    \\
		(20,3,5)    & (24,3,6)    & (28,3,7)    & (30,5,5)    \\
		(32,3,8)    & (36,3,9)    & (40,3,10)   & (42,5,7)    \\
		(44,3,11)   & (48,3,12)   & (50,9,5)    & (52,3,13)   \\
		(54,5,9)    & (56,3,14)   & (60,3,15)   & (64,3,16)   \\
		(66,5,11)   & (68,3,17)   & (70,9,7)    & (72,3,18)   \\
		(76,3,19)   & (78,5,13)   & (80,3,20)   & (84,3,21)   \\
		(88,3,22)   & (90,5,15)   & (92,3,23)   & (96,3,24)   \\
		(98,13,7)   & (100,3,25)  & (102,5,17)  & (104,3,26)  \\
		(108,3,27)  & (110,9,11)  & (112,3,28)  & (114,5,19)  \\
		(116,3,29)  & (120,3,30)  & (124,3,31)  & (126,5,21)  \\
		(128,3,32)  & (130,9,13)  & (132,3,33)  & (136,3,34)  \\
		(138,5,23)  & (140,3,35)  & (144,3,36)  & (148,3,37)  \\
		(150,5,25)  & (152,3,38)  & (154,13,11) & (156,3,39)  \\
		(160,3,40)  & (162,5,27)  & (164,3,41)  & (168,3,42)  \\
		(170,9,17)  & (172,3,43)  & (174,5,29)  & (176,3,44)  \\
		(180,3,45)  & (182,13,13) & (184,3,46)  & (186,5,31)  \\
		(188,3,47)  & (190,9,19)  & (192,3,48)  & (196,3,49)  \\
		(198,5,33)  & (200,3,50)  & (204,3,51)  & (208,3,52)  \\
		(210,5,35)  & (212,3,53)  & (216,3,54)  & (220,3,55)  \\
		(222,5,37)  & (224,3,56)  & (228,3,57)  & (230,9,23)  \\
		(232,3,58)  & (234,5,39)  & (236,3,59)  & (238,13,17) \\
		(240,3,60)  & (242,21,11) & (244,3,61)  & (246,5,41)  \\
		(248,3,62)  & (250,9,25)  & (252,3,63)  & (256,3,64)  \\
	\end{tabular}
	\caption{Triplets $(n, k, rot)$ for which our routine can produce orthogonal binary matrices.}
	\label{Table:nkrot}
\end{table}

{\bf Observation:} For certain values of even $n$, there are values of $rot$ and $k$ such that the incidence matrices that correspond to the cyclic Latin Rectangles defined by Definition \ref{Def:CyclicLatinRect} are orthogonal.

In addition to the previous observation, we would like to add one more practical (computational efficiency) goal when designing orthogonal matrices: We are interested in cyclic Latin Rectangles that give incidence matrices that are orthogonal, with as small as possible number of rows $k$. By running a simple SageMath script for even $n$ in the range $\{8, 10, 12, \ldots, 256\} $, we constructed Table \ref{Table:nkrot} for the values of $(n, k, rot)$. The entries in that table are interpreted as follows: If the triplet $(n, k, rot)$ is present in the table, then there is an efficient procedure for generating orthogonal binary matrices of size $n \times n$ from a cyclic Latin Rectangle with $k$ rows, by generating one random permutation $R_0$ with $n$ elements $\{0, 1, \ldots, n-1\}$, and by producing $k-1$ subsequent rows with rotating $R_0$ by $rot$ positions to the left. 

Next we use the following property of orthogonal matrices:
\begin{proposition}
	If $A$ and $B$ are two orthogonal matrices then their product $A\cdot B$ is an orthogonal matrix.
	\label{Proposition:ProductOfOrthogonalMatrices}
\end{proposition}

Combining Proposition \ref{Proposition:balanced} and Proposition \ref{Proposition:ProductOfOrthogonalMatrices}, it gives us a way how to construct random looking binary orthogonal matrices. For a given set of parameters $(n, k, rot)$ from Table \ref{Table:nkrot} we produce several random instances of binary orthogonal matrices, and we multiply them all together to produce a final result.

\begin{example}
	The orthogonal matrix $\mathbf{P}$ from Example \ref{Example:QuasiBinQuasiOrth} is obtained as the product $\mathbf{P} = \mathbf{M}_1 \cdot \mathbf{M}_2 \cdot \mathbf{M}_3$ from the following Latin Rectangles and their corresponding incidence matrices:
	
	$L_1=[l_{i,j}]_{3 \times 8} = \left[
	\begin{array}{cccccccc}
	6 & 5 & 4 & 3 & 1 & 7 & 0 & 2 \\
	4 & 3 & 1 & 7 & 0 & 2 & 6 & 5 \\
	1 & 7 & 0 & 2 & 6 & 5 & 4 & 3 \\
	\end{array}
	\right]$ and $\mathbf{M}_1 = \left(
	\begin{array}{cccccccc}
	0 & 0 & 1 & 0 & 1 & 0 & 1 & 0 \\
	1 & 0 & 1 & 0 & 1 & 0 & 0 & 0 \\
	0 & 0 & 0 & 1 & 0 & 1 & 0 & 1 \\
	0 & 1 & 0 & 1 & 0 & 0 & 0 & 1 \\
	1 & 0 & 1 & 0 & 0 & 0 & 1 & 0 \\
	0 & 1 & 0 & 0 & 0 & 1 & 0 & 1 \\
	1 & 0 & 0 & 0 & 1 & 0 & 1 & 0 \\
	0 & 1 & 0 & 1 & 0 & 1 & 0 & 0 \\
	\end{array}
	\right),$
	
	$L_2=[l_{i,j}]_{3 \times 8} = \left[
	\begin{array}{cccccccc}
	3 & 6 & 1 & 7 & 4 & 2 & 0 & 5 \\
	1 & 7 & 4 & 2 & 0 & 5 & 3 & 6 \\
	4 & 2 & 0 & 5 & 3 & 6 & 1 & 7 \\
	\end{array}
	\right]$ and $\mathbf{M}_2 = \left(
	\begin{array}{cccccccc}
	0 & 0 & 1 & 0 & 1 & 0 & 1 & 0 \\
	1 & 0 & 1 & 0 & 0 & 0 & 1 & 0 \\
	0 & 1 & 0 & 1 & 0 & 1 & 0 & 0 \\
	1 & 0 & 0 & 0 & 1 & 0 & 1 & 0 \\
	1 & 0 & 1 & 0 & 1 & 0 & 0 & 0 \\
	0 & 0 & 0 & 1 & 0 & 1 & 0 & 1 \\
	0 & 1 & 0 & 0 & 0 & 1 & 0 & 1 \\
	0 & 1 & 0 & 1 & 0 & 0 & 0 & 1 \\
	\end{array}
	\right),$
	
	$L_3=[l_{i,j}]_{3 \times 8} = \left[
	\begin{array}{cccccccc}
	4 & 2 & 5 & 6 & 3 & 1 & 0 & 7 \\
	5 & 6 & 3 & 1 & 0 & 7 & 4 & 2 \\
	3 & 1 & 0 & 7 & 4 & 2 & 5 & 6 \\
	\end{array}
	\right]$ and $\mathbf{M}_3 = \left(
	\begin{array}{cccccccc}
	0 & 0 & 1 & 0 & 1 & 0 & 1 & 0 \\
	0 & 1 & 0 & 1 & 0 & 1 & 0 & 0 \\
	0 & 1 & 0 & 0 & 0 & 1 & 0 & 1 \\
	1 & 0 & 1 & 0 & 1 & 0 & 0 & 0 \\
	1 & 0 & 0 & 0 & 1 & 0 & 1 & 0 \\
	1 & 0 & 1 & 0 & 0 & 0 & 1 & 0 \\
	0 & 1 & 0 & 1 & 0 & 0 & 0 & 1 \\
	0 & 0 & 0 & 1 & 0 & 1 & 0 & 1 \\
	\end{array}
	\right).$  
\end{example}

The only remaining part is to find a way how to avoid a direct classical matrix multiplication procedure, but still to construct an orthogonal matrix that is a product of two orthogonal matrices. We achieve that goal with the following Proposition:

\begin{proposition}
	Let $L_1=[l_{i,j}]_{k \times n}$ with its corresponding columns $C^{(1)}_0, C^{(1)}_1, \ldots, C^{(1)}_{n-1}$ and $L_2=[l_{i,j}]_{k \times n}$ with its corresponding columns $C^{(2)}_0, C^{(2)}_1, \ldots, C^{(2)}_{n-1}$ are two Latin Rectangles, and let $\mathbf{M}_1$ and $\mathbf{M}_2$ are their corresponding incidence matrices. 
	
	Let define $n$ initial sets $D^{(0)}_i = \{i\}, i \in \{0,\ldots,n-1\}$, let the operator $\bigoplus$ denote exclusive union of sets and let us define $D^{(j)}_i$ with the following recursive relations:
	\begin{equation}
	D^{(j)}_i = \bigoplus_{l\in C^{(j)}_i} D^{(j-1)}_l, \ \ \text{for\ \ } i=0, 1, \ldots n-1,\ \ \text{and\ \ } j=1, 2
	\label{Equation:FastMatrixMultiplication}
	\end{equation} 
	Then, the sets $D^{(2)}_i, i \in \{0,\ldots,n-1\}$ are the support sets for the columns of the product matrix: $\mathbf{P} = \mathbf{M}_1 \cdot \mathbf{M}_2$.
	\label{Proposition:FastMatrixMultiplication}
\end{proposition}
\begin{proof}
	(Sketch) The correctness of relation (\ref{Equation:FastMatrixMultiplication}) comes from Definition \ref{Definition:incidencematrix} for incidence matrices and the standard definition of matrix multiplication applied to matrices $\mathbf{M}_1$ and $\mathbf{M}_2$.
\end{proof}

We summarize all definitions, propositions and observations in the Theorem \ref{Theorem:cd}.

\begin{theorem}
	Let $\mathbf{P}_{a, b}$ be a quasi-binary and quasi-orthogonal $n\times n$ matrix with even $n$ over a finite field $\mathbb{F}_{2^m}$ of characteristic $2$ and let  $\mathbf{P}_{c, d}$ is its inverse matrix i.e. $\mathbf{P}^{-1}_{a, b} = \mathbf{P}_{c, d}$. Then it holds that:
	\begin{equation}
	\begin{cases}
	c = & \dfrac{a}{a^2 + b^2}\\
	d = & \dfrac{b}{a^2 + b^2}
	\end{cases}.
	\label{Equation:abcd}
	\end{equation}
	\label{Theorem:cd}
\end{theorem}

\begin{proof}
	Since $n$ is even, we can use Proposition \ref{Proposition:odd} and by multiplying the $i$--th row with the $i$--th column we have that $a c + b d = 1 $, while when multiplying the $i$--th row with the $j$--th column ($i \ne j$) we have that $ b c + a d = 0 $. Thus we obtain the following system of two equations:
	\begin{equation}
	\begin{cases}
	a c + b d = & 1\\
	b c + a d = & 0
	\end{cases}.
	\label{Equation:InitialSystem}
	\end{equation}
	Knowing $c$ and $d$ we get the following solution:
	\begin{equation}
		\begin{cases}
			a = & \dfrac{c}{c^2 - d^2}\\
			b = & \dfrac{b}{d^2 - c^2}
		\end{cases}.
		\label{Equation:SoultionGeneral}
	\end{equation}
	Since the field is $\mathbb{F}_{2^m}$ of characteristic $2$, Equation (\ref{Equation:abcd}) follows directly.
\end{proof}

As a result of the previous analysis we designed  the algorithm $\mathtt{RandomOrthogonalBinaryMatrix}(n)$ for generation of binary orthogonal matrices of size $n\times n$ (given in Table \ref{Alg:RandomOrthogonalBinaryMatrix}). Note that we chose 6 iterations in Step 3 of the algorithm $\mathtt{RandomOrthogonalBinaryMatrix}(n)$. That is an arbitrary value obtained by our concrete experiments, that achieves the average Hamming weight of the columns and rows in the produced orthogonal matrix to be distributed around $\frac{n}{2}$. Further research is needed to analyze the possibility this number to be less than 6, and still to have a statistical distribution of the Hamming weight of the columns and rows to be around $\frac{n}{2}$.

\begin{table}[!]
	\centering
	\normalsize
	\begin{tabular}{|c|}
		\hline
		\parbox{7.5cm}{\vspace{0.1cm} \centering \bf $\mathtt{RandomOrthogonalBinaryMatrix}$ \vspace{0.1cm}}\\
		\hline
		% \> for next tab, \\ for new line...
		\parbox{7.5cm}{\vspace{0.1cm}{\bf Input:} $n$,\\
			{\bf Output:} A random orthogonal $n\times n$ binary matrix $\textbf{P}$. 
		\vspace{0.1cm}} \\
		\hline
		\parbox{7.5cm}{
			\begin{enumerate}
				\item For the given $n$ find the corresponding triplet $(n, k, rot)$ from the Table \ref{Table:nkrot}. If the triplet is not present, Return $\mathtt{Error}$
				\item Initialize $n$ sets $D^{(0)}_i = \{i\}, i \in \{0,\ldots,n-1\}$ as in Proposition \ref{Proposition:FastMatrixMultiplication}
				\item For $j=1$ to 6 do 
				\begin{itemize}
					\item Generate a random cyclic Latin Rectangle $L_j$ of size ${k \times n}$ with columns $C^{(j)}_0, C^{(j)}_1, \ldots, C^{(j)}_{n-1}$
					\item $D^{(j)}_i = \bigoplus_{l\in C^{(j)}_i} D^{(j-1)}_l, \ \ \text{for\ \ } i=0, 1, \ldots n-1$
				\end{itemize}
				\item Set a binary matrix $\textbf{P}$ of size $n\times n$ for which the sets $D^{(6)}_i, i \in \{0,\ldots,n-1\}$ are the support sets for its columns
				\item Return $\textbf{P}$.
			\end{enumerate}
		}\\
		\hline
	\end{tabular}
	\caption{The algorithm for fast construction of orthogonal binary matrices of size $n\times n$. }\label{Alg:RandomOrthogonalBinaryMatrix}
\end{table}

We have implemented our algorithm in several programming languages such as C, SageMath (which is using a variant of Python) and in Mathematica. An example of the Mathematica output of a randomly generated $16\times 16$ quasi-binary and quasi-orthogonal matrix is given in Figure \ref{Figure:Random16x16} and Figure \ref{Figure:Random16x16Inv}.

\begin{figure}[h!]
	\centering
	\includegraphics[scale=0.5]{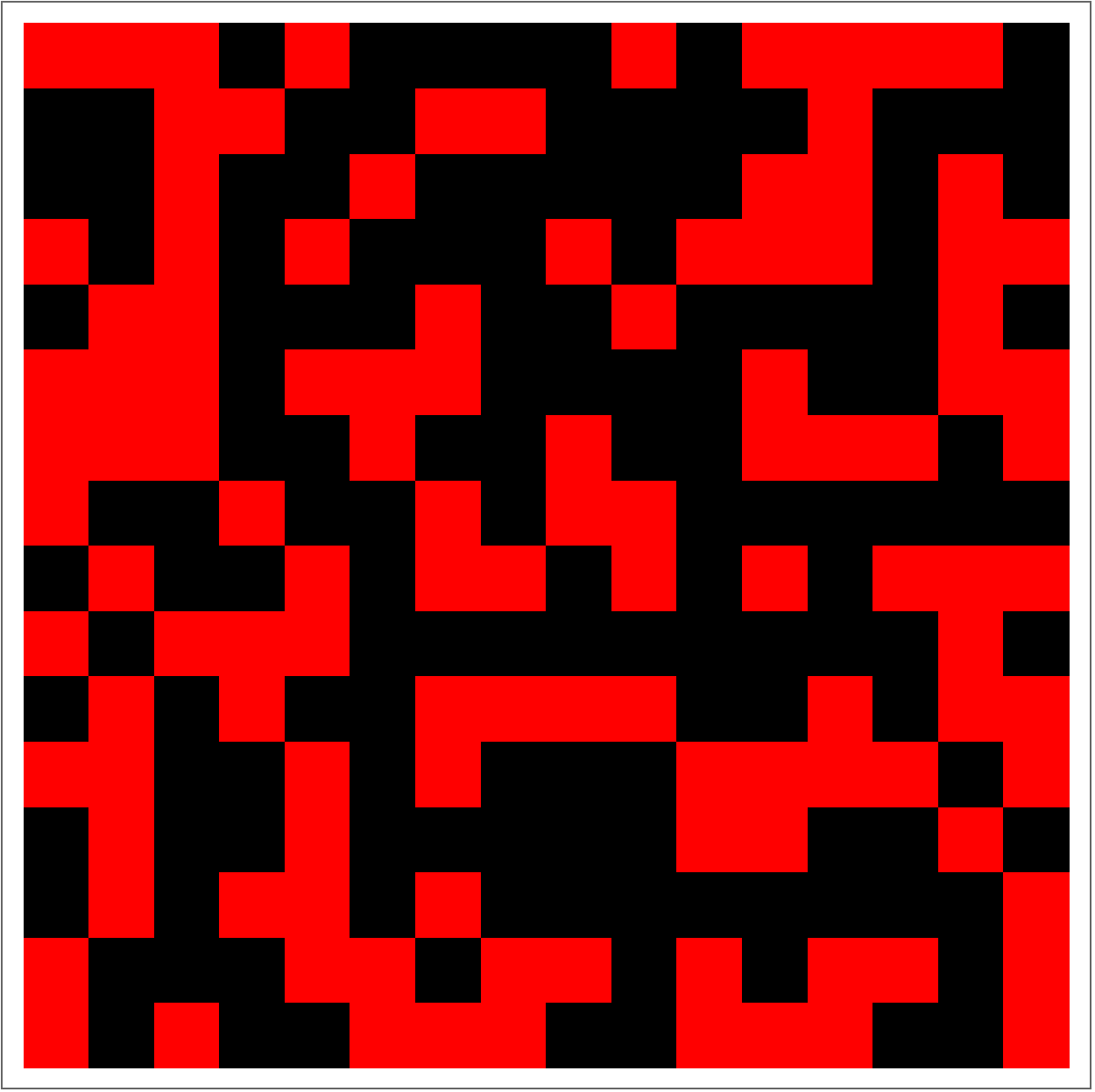}
	\caption{A random $16\times 16$ quasi-binary and quasi-orthogonal matrix $\mathbf{P}_{a, b}$ produced by the Algorithm $\mathtt{RandomOrthogonalBinaryMatrix}$ implemented in Mathematica. Note that the red squares correspond to the values of $a$ and the black squares correspond to $b$ where $a, b \in \mathbb{F}_{2^m}$.}
	\label{Figure:Random16x16}
\end{figure}

\begin{figure}[h!]
	\centering
	\includegraphics[scale=0.5]{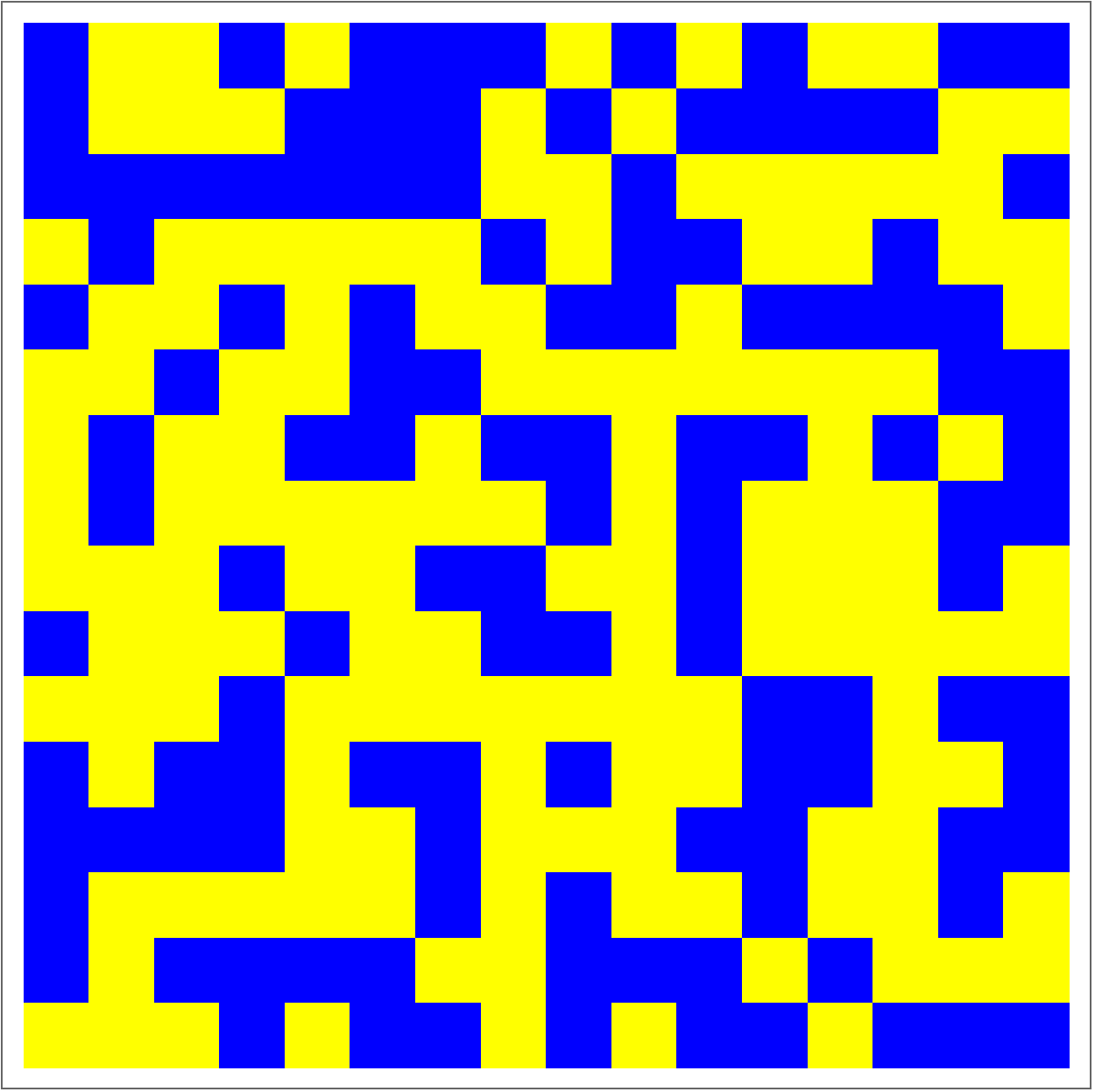}
	\caption{The corresponding inverse $16\times 16$ matrix $\mathbf{P}_{c, d}$. The blue squares correspond to the values of $c$ and the yellow squares correspond to $d$ where $a, b, c$ and $d$ are connected with the relation (\ref{Equation:abcd}).}
	\label{Figure:Random16x16Inv}
\end{figure}

\section{Conclusions}
\label{sec:conclusions}
We defined a new family of matrices over finite fields that we named Quasi-Binary and Quasi-Orthogonal Matrices. The reasons why we call those matrices quasi-binary are due to the fact that they have only two elements $a, b \in \mathbb{F}_q$, but those elements are not $0$ and $1$. The reasons why we call them quasi-orthogonal are due to the fact that their inverses are obtained not just by a simple transposition, but there is a need for an additional operation: a replacement of $a$ and $b$ by two other values $c$ and $d$. We gave a simple relation between $a, b, c$ and $d$ for any finite field, in particular for a finite field with characteristic 2. 

Our construction is based on incident matrices from cyclic Latin Rectangles and our algorithm is very efficient since it completely avoids operations of matrix-matrix or matrix-vector multiplications. We have implemented it in several computer languages such as C, SageMath (Python) and Mathematica. 

%\vspace{-0.5cm}
\bibliographystyle{IEEEtran}

%\vspace{-0.2cm}
\bibliography{IEEEabrv,bib}

\end{document}